\documentclass[conference]{IEEEtran}

\if CLASSOPTIONcompsoc
  \usepackage[nocompress]{cite}
\else
  \usepackage{cite}
\fi

\usepackage[utf8x]{inputenc}
\usepackage{listings}
\usepackage{color}
\usepackage{epsfig}
\usepackage{url}
\usepackage{amsmath}
\usepackage{cite}
\usepackage{graphicx}
\usepackage{color}
\usepackage[bookmarks=false]{hyperref}
\usepackage{amssymb,amsmath,dsfont}
\usepackage{amsthm}
\usepackage[linesnumbered,lined, ruled]{algorithm2e}
\usepackage{epstopdf}
\usepackage{lipsum}
\usepackage{bbm}
\usepackage{tablefootnote}

\usepackage{tikz}
\usetikzlibrary{automata,arrows,positioning,calc}

\usepackage{multirow}

\usepackage{breqn}
\SetAlFnt{\footnotesize}

\usepackage{graphics}

\newtheorem{theorem}{Theorem}

\newtheorem{lemma}{Lemma}

\theoremstyle{remark}
\newtheorem*{remark}{Remark}

\usepackage{csquotes}
\usepackage{amsmath}
\usepackage{amssymb}
\usepackage{subcaption}
\usepackage{graphicx}
\usepackage{accents}
\usepackage{mathtools}
\usepackage[normalem]{ulem}

\SetCommentSty{mycommfont}

\newcommand{\eqdef}{\vcentcolon=}

\graphicspath{figures/}

\usepackage{mathtools}

\SetKwRepeat{Repeat}{repeat}{until}
\SetKwRepeat{Forever}{repeat}{forever}
\SetKwRepeat{On}{on}{end}
\SetNlSty{bfseries}{\color{black}}{}

\begin{document}

\title{Improved Delay Bound for a Service Curve Element with Known Transmission Rate}


\author{
	\IEEEauthorblockN{Ehsan Mohammadpour, Eleni Stai, Jean-Yves Le Boudec\\
	}
\IEEEauthorblockA{\'Ecole Polytechnique F\'ed\'erale de Lausanne, Switzerland\\
$\{$firstname.lastname$\}$@epfl.ch}}

\maketitle

\begin{abstract}
Network calculus is often used to prove delay bounds in deterministic networks, using arrival and service curves. We consider a FIFO system that offers a rate-latency service curve and where packet transmission occurs at line rate without pre-emption. The existing network calculus delay bounds take advantage of the service curve guarantee but not of the fact that transmission occurs at full line rate. In this letter, we provide a novel, improved delay bound which takes advantage of these two features. Contrary to existing bounds, ours is per-packet and depends on the packet length. We prove that it is tight.
\end{abstract}

\IEEEpeerreviewmaketitle

\section{Introduction}
\label{sec:intro}
In the context of deterministic networking \cite{detnet} or time-sensitive networking \cite{TSN}, delays at network elements have to be bounded in the worst case, not in average. Computing and formally verifying delay bounds is often done by using network calculus \cite[Section 1.4]{le_boudec_network_2001}. For a FIFO network element, this involves two steps. First, an arrival curve, say $\alpha$, is formulated for the aggregated input traffic. Specifically, $\alpha(t)$ is an upper bound on the number of bits that may be submitted by the traffic of interest into the network element within any $t$ time units. The function $\alpha$ depends on the knowledge of the applications that generates the traffic and on the speed at which data can arrive the network element. Second, the details of the inner workings of the network element are abstracted by using a service curve, say $\beta$ (also called ``minimum" service curve). This service curve is typically a rate-latency function, i.e., of the type $\beta(t)=\max(0, R(t-T))$, where $R$ (the rate) and $T$ (the latency) are fixed parameters that are specific to the network element and to the traffic class. An exact definition of service curve can be found in \cite[Section 1.3]{le_boudec_network_2001}, \cite[Section 2.3]{Changbook}, and is recalled in Section~\ref{sec:delay}, Eq.~\eqref{eq:min_service_def}.
Roughly speaking, such a rate-latency service curve means that the input traffic is guaranteed to receive a service rate at least equal to $R$, except for possible service interruptions that may impact the delay by at most $T$ units of time.

Then, a delay bound given by network calculus is the horizontal deviation between the arrival and service curves \cite[Section 3.1.11]{le_boudec_network_2001}, which in this case is
\begin{align}
 \Delta = T+ \sup_{t\geq 0}\left\{ \frac{\alpha(t)}{R}-t \right\}.\label{eq:net_calc_bound}
\end{align}
(In the above formula, $\Delta$ is finite if the supremum is finite, otherwise $\Delta$ is infinite).
This methodology has been successfully applied to many network elements involving a variety of schedulers such as priority schedulers \cite[Chapter~7]{bouillard2018deterministic}, all schedulers that fall in the class of guaranteed rate scheduling \cite{GLV95}, \cite[Section 2.1]{le_boudec_network_2001} (including the widespread deficit round robin scheduler \cite[Chapter~8]{bouillard2018deterministic}), and more recently Audio-Video Bridging \cite{azua_complete_2014} and the Credit Based Shaper \cite{mohammadpour_latency_2018,daigmorte_modelling_2018,zhao_timing_analysis_2018}.

The bound in Eq. \eqref{eq:net_calc_bound} is tight if the only information available is the arrival curve $\alpha$ and the service curve $\beta$. However, in all of the examples we just mentioned, there is an additional information that is not taken into account by Eq. \eqref{eq:net_calc_bound} Specifically, packet transmission occurs at the physical line rate $c$, which is often much larger than the rate $R$ guaranteed by the rate-latency service curve. For example, with a Deficit Round Robin (DRR) scheduler that handles $n$ classes of traffic of equal importance, for every class the rate $R$ is equal to $c/n$. Using both the information of the rate-latency service curve, $\beta$, and of the constant transmission rate, $c$, to compute a delay bound is not straightforward. This was not done before, and is the contribution of this paper.

In this paper, we exploit the information on the transmission rate to provide a bound on the delay at a FIFO system that improves on the network calculus bound in (\ref{eq:net_calc_bound}). Specifically, the bound is per-packet and depends on the packet length. We reached this improved bound by combining the min-plus representation of service curve and max-plus representation of arrival curve \cite{le_boudec_theory_2018}. We show that the bound is tight, at least when the arrival curve is concave.
\section{System Model}\label{sec:sys}
We consider a FIFO system with a queue and a transmission subsystem, as in Fig. \ref{fig:fifo_node}. Upon arrival, packets enter the queue and are stored in FIFO order. A scheduler decides when the packet at the head of the queue is selected for transmission. The scheduler typically arbitrates between this queue and other queues (not shown), therefore the packet at the head of the queue have to wait even if there is no packet of this queue in transmission. When the packet at the head of this queue is selected for transmission, it is transmitted at a constant rate $c$ until it is completely transmitted, i.e., there is no pre-emption. Let $A_n$ be the arrival time of the packet $n$, where the numbering of packets is by order of arrivals, and let $Q_n$ be the time at which packet $n$ is selected for transmission. The FIFO assumption means that $Q_n \leq Q_{n+1}$. Let $l_n$ be the length of packet $n$ and $L_{\max}$ the maximum packet length. The packet $n$ leaves the system at time $D_n = Q_n +\frac{l_n}{c}$. We call $D_n-A_n$ the ``response time" of packet $n$.

Furthermore, we assume that the scheduler is such that the complete system offers to the total flow of all incoming packets a service curve $\beta(t)=\max(0, R(t-T))$ with $R\leq c$. In many cases, the rate $R$ is much less than $c$; this occurs for example when the transmission capacity $c$ is shared between this FIFO system and other subsystems dedicated to other classes of traffic, as in \cite{daigmorte_modelling_2018}.

We also assume that the total flow of all incoming packets is packetized, i.e., we consider that all bits of packet $n$ arrive at the same time instant $A_n$. Furthermore, we assume that the total flow of all incoming packets is constrained by an arrival curve $\alpha$.

\begin{figure}
	\centering
	\includegraphics[width=0.7 \linewidth]{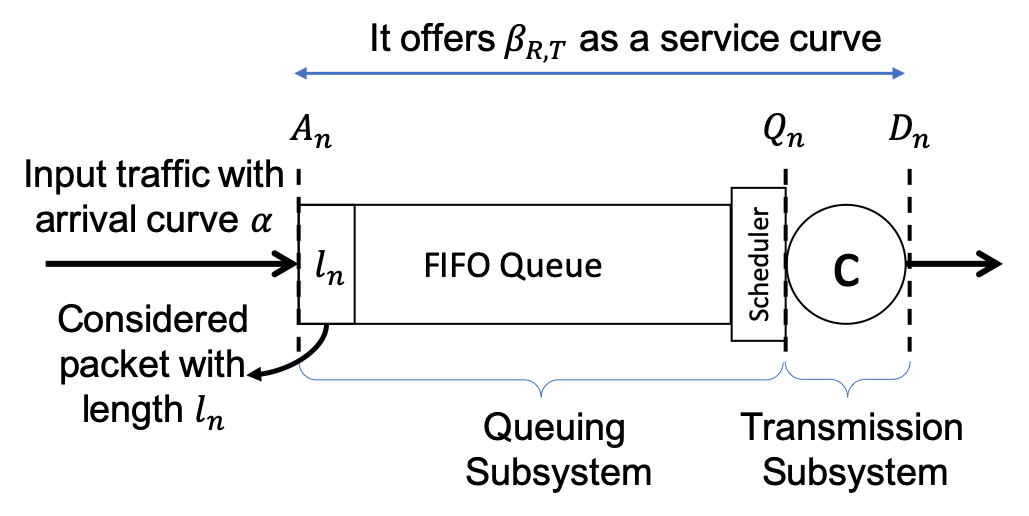}
	\caption{Model of the considered FIFO system.}
	\label{fig:fifo_node}
\end{figure} 
\section{Improved Delay Bound}\label{sec:delay}
In this section, we derive a delay bound for a general FIFO system as described in Section~\ref{sec:sys}.
\begin{theorem}\label{thm:response_fifo}
	(Upper bound on the response time at a FIFO system)
	Consider a FIFO system as in Section~\ref{sec:sys}, i.e., one that offers a rate-latency service curve with parameters ($R$,$T$), and where, as soon as a packet starts to be transmitted, it is transmitted at a constant rate $c\geq R$. Assume that the total input is packetized and has an arrival curve $ \alpha$. For a packet of length $l$, the response time is upper bounded by:
	\begin{align} \label{eq:fifo_response_time_arr_con}
	\Delta_l= \Delta- l \bigg(\frac{1}{R} - \frac{1}{c} \bigg),
	\end{align} where $\Delta$ is the network calculus bound, given in Eq. \eqref{eq:net_calc_bound}.
\end{theorem}
\begin{proof}
We use the notation in Section~\ref{sec:sys} and call $l_{k}$ the length of the $k^{th}$ packet with $k=1,2...$. Now, let $n$ be the index of the packet of interest, which has length $l$,  i.e., $l_n=l$. By Lemma \ref{lemma:min_service_curve} there exists an $m\leq n$ such that:
	\begin{align}\label{eq:scheduler_del_4nm}
	\beta(Q_n-A_m) \leq \sum_{k=m}^{n-1} l_k,
	\end{align}
where $\beta(t)=\max(0, R(t-T))$ is the rate-latency service curve.
	Let $\beta^{\uparrow}$ be the upper pseudo-inverse of $\beta$, defined by $\beta^{\uparrow}(t) =\sup \{s\geq 0 | \beta(s)\leq t\}=\inf \{s\geq 0 | \beta(s)>t\}$. By \cite[Section 10.1]{liebeherr_duality_2017}, for a wide-sense increasing function $F(.)$, $F(x)\leq y \Rightarrow x \leq F^{\uparrow}(y)$. Then, from Eq. \eqref{eq:scheduler_del_4nm}, $Q_n-A_m \leq \beta^{\uparrow}\Big(\sum_{k=m}^{n-1} l_k\Big)$.
	Therefore, $Q_n$ satisfies,
	\begin{align}\label{eq:scheduler_del_6}
	Q_n \leq \max_{m\leq n} \Bigg\{A_m+\beta^{\uparrow}\Big(\sum_{k=m}^{n-1} l_k\Big)\Bigg\}.
	\end{align}
The input traffic has an arrival curve $\alpha$, thus by the max-plus representation of arrival curves in \cite[Theorem 1]{le_boudec_theory_2018}:
	\begin{align}\label{eq:scheduler_del_8}
	\sum_{k=m}^{n} l_k \leq \alpha^+(A_n-A_m),
	\end{align}
		\noindent where $\alpha^+$ is the right-limit of $\alpha$. By excluding the packet of interest from the summation in Eq. \eqref{eq:scheduler_del_8}, we obtain:
	\begin{align}\label{eq:scheduler_del_9-arr-2}
	\sum_{k=m}^{n-1} l_k \leq \alpha^+(A_n-A_m) - l.
	\end{align}
	By using Eq. (\ref{eq:scheduler_del_9-arr-2}) in Eq. (\ref{eq:scheduler_del_6}), we have,
	\begin{align}
	Q_n \leq \max_{m\leq n} \Bigg\{&A_m+\beta^{\uparrow}\bigg(\alpha^+(A_n-A_m) - l\bigg)\Bigg\}.
	\end{align}
	By defining $t \eqdef A_n-A_m \geq 0$, we further obtain,
	\begin{align}
	Q_n - A_n \leq \sup_{t \geq 0} &\Bigg\{-t+\beta^{\uparrow}\bigg(\alpha^+(t)- l\bigg)\Bigg\}.
	\end{align}	
Note that $\beta^{\uparrow} (x)=\frac{x}{R} +T$, therefore
	\begin{align}\label{eq:intermpart2}
	Q_n - A_n &\leq \sup_{t \geq 0} \Bigg\{-t+\frac{\alpha^+(t)- l}{R}+T\Bigg\} \nonumber
	\\
	&=T - \frac{l}{R} + \frac{1}{R}\sup_{t \geq 0} \Bigg\{\alpha^+(t)-Rt\Bigg\} .
	\end{align}
By applying Lemma \ref{lemma:right_lim_arr} in Eq. \eqref{eq:intermpart2}, we have:
	\begin{align}\label{eq:intermpart3}
	\nonumber Q_n - A_n &\leq T - \frac{l}{R} + \frac{1}{R}\sup_{t \geq 0} \Bigg\{\alpha(t)-Rt\Bigg\} = \Delta-\frac{l}{R},
	\end{align}
where $\Delta$ is the network calculus bound, given in (\ref{eq:net_calc_bound}). Now observe that $D_n=Q_n+\frac{l}{c}$, which concludes the proof. \end{proof}

	\begin{lemma}\label{lemma:min_service_curve}
		If a FIFO system has (i) $\beta$ as a service curve, and (ii) packetized input, then for every packet $n$, there exists a packet index $m \leq n $ such that, $\beta(Q_n - A_m) \leq \sum_{k = m}^{n-1} l_{k}$, where $l_k$ is the length of the $k^{th}$ packet.
	\end{lemma}
	\begin{proof}
		\noindent The network calculus framework uses cumulative functions $I(t)$ and $O(t)$, where $I(t)$ is the number of bits that have arrived up to (excluding) time $t$ and $O(t)$ is the number of bits that have been served up to (excluding) time $t$. Then \cite{le_boudec_network_2001, Changbook} the system is said to offer the flow a service curve $\beta$ if for any $t\geq 0$, there exists an $s\in[0,t]$ such that
		\begin{equation}
		O(t)~\geq~I(s)~+~\beta(t-s). \label{eq:min_service_def}
		\end{equation}
		Let $n$ be some packet index. We have $O(Q_n) = \sum_{i=1}^{n-1} l_i \label{eq:output_t}$, since $Q_n$ is the time at which packet $n$ starts being transmitted and, by the FIFO property, all packets before $n$ have been served by that time.
		
		Now apply \eqref{eq:min_service_def} with $t=Q_n$. For the resulting $s$, let $m$ be the smallest packet index such that, $s \leq A_{m}$ and thus $I(s) = \sum_{i=1}^{m-1} l_i \label{eq:beta_m}$, with the convention that an empty sum is equal to $0$ (which occurs when $m=1$).
		Note that, at this point, we do not know if $n\leq m$ or if $n>m$ (we know that $s\leq Q_n$, but, this does not imply that $m\leq n$, for example it is quite possible that $A_{n+1} < Q_n$ since packets may wait in the queue).  But, in any case, since $s\leq A_m$ and $\beta$ is wide-sense increasing:
		\begin{equation}
		\beta(Q_n-s) \geq \beta(Q_n-A_m). \label{eq:beta_rel}
		\end{equation}
		
		By using \eqref{eq:beta_rel} and the expressions of $O(Q_n)$ and $I(s)$ in \ref{eq:min_service_def}, we obtain:		
		\begin{equation}
		\sum_{i=1}^{n-1} l_i \geq~\sum_{i=1}^{m-1} l_i + \beta(Q_n - A_m).
		\end{equation}
		
		If $m > n$, we obtain $\beta(Q_n - A_m) < 0$, which is a contradiction; therefore, $m\leq n$.
	\end{proof}
	\begin{lemma}\label{lemma:right_lim_arr}
		If $f(.)$ is a wide-sense increasing function and $f^+(.)$ is its right-limit, then for any $R>0$:
		\begin{equation}
		\sup_{t \geq 0}\Big(f^+(t)-Rt\Big) = \sup_{t \geq 0}\Big(f(t)-Rt\Big).
		\end{equation}	\end{lemma}
	\begin{proof}
		Let $K=\sup_{t \geq 0}\Big(f(t)-Rt\Big)$ and $K'=\sup_{t \geq 0}\Big(f^+(t)-Rt\Big)$. We want to prove that $K=K'$. To do so, first we show that $K\leq K'$; and second that $K \geq K'$.
		
		$K\leq K'$: The function $f$ is wide-sense increasing; therefore for any $t\geq 0$, we have:
		\begin{align}
		f(t) \leq f^+(t) \implies f(t)-Rt \leq f^+(t)-Rt .
		\label{eq:fimplic}
		\end{align}
		Using \eqref{eq:fimplic}, it is trivially shown that $K \leq K'$.
		
		$K \geq K'$: The function $f$ is wide-sense increasing; therefore for any $t\geq 0$ and $\varepsilon > 0$, we have $f^+(t) \leq f(t+\varepsilon)$; thus:
\begin{align}
		\nonumber  f^+(t)-Rt &\leq f(t+\varepsilon)-Rt \\
		\nonumber & = f(t+\varepsilon)-R(t+\varepsilon)+ R\varepsilon\\
		& \leq \sup_{u \geq 0}\Big(f(u)-R(u)\Big)+ R\varepsilon
		= K + R\varepsilon,\label{eq:flesss}
\end{align}
		i.e., $K + R\varepsilon$ is an upper bound on $f^+(t)-Rt$.
			By definition, $K'$ is the lowest such upper bound. Thus $K' \leq K + R\varepsilon$. This holds for any $\varepsilon > 0$, thus $K' \leq K$.
	\end{proof}
\begin{remark}
If the input in the FIFO system consists of multiple flows, then the arrival curve, $\alpha(t)$, is an envelope for the aggregate of the flows. If the flows have different minimum packet lengths, then Theorem \ref{thm:response_fifo} provides a distinct delay bound for every flow, $\Delta_{L_{\min}^f}$, where $L_{\min}^f$ is the minimum packet length of flow $f$.
\end{remark}
\textbf{Case Studies:} Hereafter we provide two examples that illustrate the improvement of Theorem \ref{thm:response_fifo} over the network calculus delay bound.

In the first example, we compute delay bounds for two Audio-Video Bridging (AVB) classes in a TSN scheduler \cite{mohammadpour_latency_2018,daigmorte_modelling_2018}. We consider the traffic specification in \cite{zhao_timing_analysis_2018}. Also, the idle slopes of the TSN scheduler for classes A and B are set the same as \cite{zhao_timing_analysis_2018} and are respectively equal to $60\%$ and $15\%$ of the link rate. In addition, we use the rate-latency service curves of Theorem 1 in \cite{mohammadpour_latency_2018}. For a link rate $c=1$ Gbps, we find that the delay bound improvement for a packet of class A with length $1.499$ KB is $8\mu$s per hop. Also, the improvement for a packet of class B with length $1.438$ KB is $66$ $\mu$s per hop. When the link rate is $100$Mbps, the improvement for the packet of class A is $98~\mu$s and for the packet of class B is $736$ $\mu$s per hop. For both link rates the delay bound of Theorem \ref{thm:response_fifo} improves the network calculus bound for classes A and B by around $2\%$ and $10\%$, respectively. The improvement is small but non-negligible. 

In the second example, we consider a node with per-flow queuing and DRR arbitration policy, where $n$ flows share a link with rate $c$. We assume that all flows have the same maximum packet length $L$ and the same quantum value, $Q=L$. Therefore, the rate-latency service curve parameters for all flows are the same given by (Section 9.2.3 of \cite{bouillard2018deterministic}):
{\begin{align}
R&= \frac{Q}{\sum_{j=1}^{n}Q} c= \frac{c}{n},\\
\nonumber T &= {\frac{\sum_{j=1, j\neq i}^{n}(L +Q)}{c}} + L \Big(\frac{1}{R}-\frac{1}{c}\Big) = \frac{3L}{c} (n-1).
\end{align}}
Assume that the maximum burstiness of each flow is limited by its maximum packet length, i.e., $L$. Then, the network calculus bound is $\Delta=\frac{(4n-3)L}{c}$ and the improvement in the delay bound for a packet with length $L$ is $\frac{(n-1)L}{c}$, approximately $25\%$, which is significant.

\section{Tightness of the Improved Delay Bound}\label{sec:tight}
In this section, we prove that Theorem~\ref{thm:response_fifo} is tight for a concave arrival curve $\alpha$. Observe that, since the input is packetized, $\alpha^+(0)$ is an upper bound on the length of any packet (where $\alpha^+$ is the right-limit of $\alpha$). Therefore, we need to assume that $\alpha^+(0)\geq L_{\max}$, where $L_{\max}$ is the maximum packet length. Also, the bound in Theorem~\ref{thm:response_fifo} is of interest only when the network calculus bound is finite. 

\begin{theorem}\label{thm:response_fifo_tight}
	If the network calculus bound in Eq. \eqref{eq:net_calc_bound} is finite, and the arrival curve $\alpha$ is concave such that $\alpha^+(0)\geq L_{\max}$, the bound of Theorem \ref{thm:response_fifo} is tight.
	
	Specifically, consider: a rate $R$, a latency $T$, a maximum packet length $L_{\max}$, a concave arrival curve $\alpha$ such that $\alpha^+(0)\geq L_{\max}$ and the bound in Eq. \eqref{eq:net_calc_bound} is finite, a transmission rate $c\geq R$, and a packet length $l\leq L_{\max}$. There exists a FIFO system where the input is packetized and has arrival curve $\alpha$, which offers a rate-latency service curve guarantee with rate $R$ and latency $T$, and in which there is an execution trace where a packet of length $l$ has a delay equal to \eqref{eq:fifo_response_time_arr_con}.
	\end{theorem}
\begin{proof}	
\textbf{Step 1}.
	We construct a simulation trace.
	
	(a) For the input, we first determine a time instant $t'$ that achieves the network calculus bound $\Delta$ in \eqref{eq:net_calc_bound}.
	Since the arrival curve $\alpha$ is concave, it is continuous, except perhaps at $t=0$. Therefore, there are two cases for $\Delta$: either $\Delta = T+  \frac{\alpha(t')}{R}-t'$ for some $t'>0$ or $\Delta = T+  \frac{\alpha^+(0)}{R}$. In the first case, due to continuity at $t'$,  $\alpha^+(t')=\alpha(t')$; in both cases there is some $t'$ such that:
\begin{equation}\label{eq:delta_tight}
\Delta = T+  \frac{\alpha^+(t')}{R}-t'.
\end{equation}
Then we construct the function $\tilde{I}(t) = \min(\alpha(t), \alpha^+(t'))$, which represents a fluid input that has $\alpha$ as arrival curve and delivers $\alpha^+(t')$ bits in total (the gray line in Fig. \ref{fig:delay_curves}).
\begin{figure}[t]
		\centering
		\includegraphics[width=0.9 \linewidth]{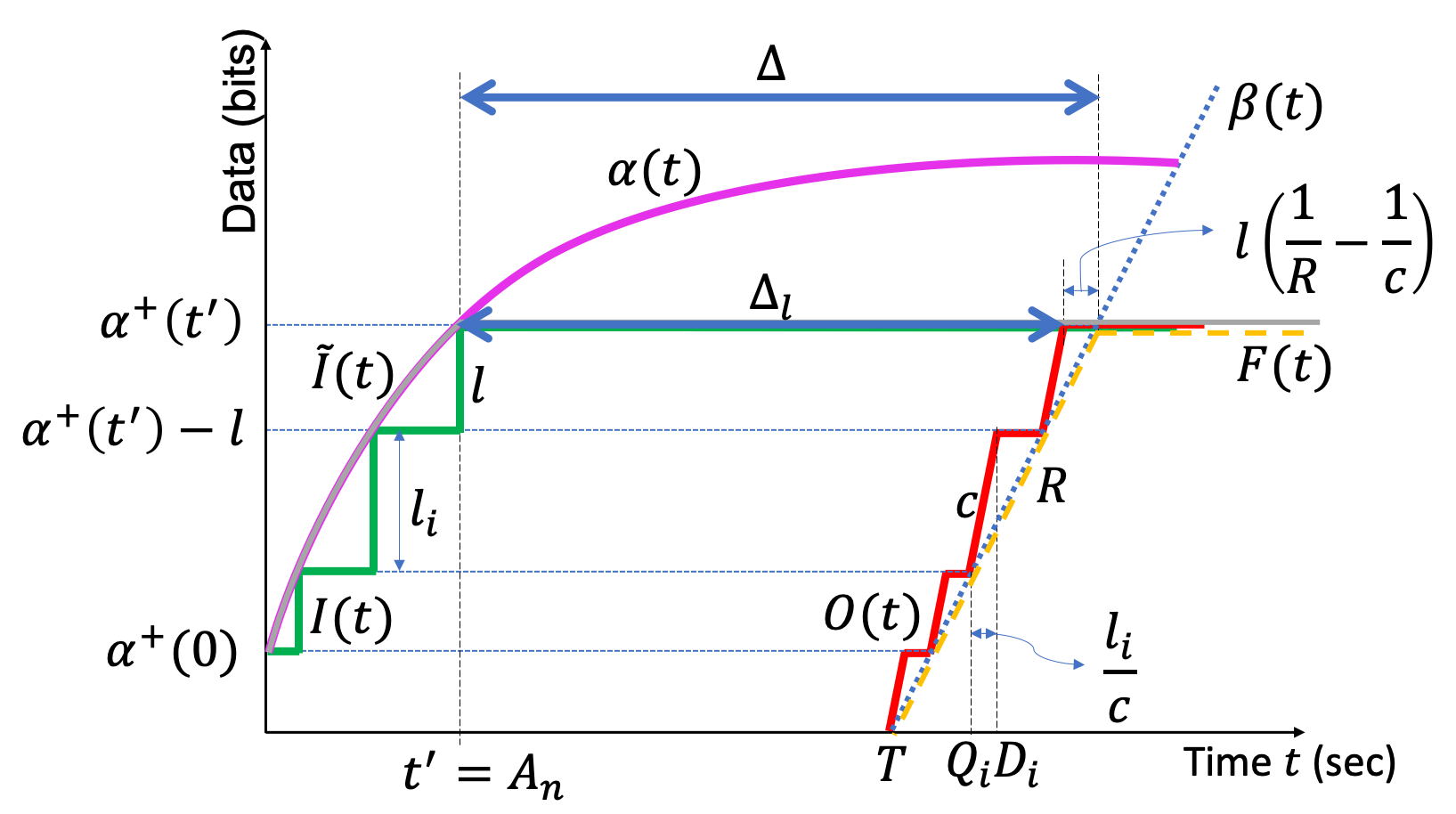}
		\caption{The execution trace used in the proof of Theorem~\ref{thm:response_fifo_tight}. The delay bound of the packet with length $l$ that arrives at time $t'$ is $\Delta_l$.}
		\label{fig:delay_curves}
	\end{figure}

(b) Since the system should have packetized input, we transform $\tilde{I}(t)$ into a train of packets.
	We determine the packet length sequence $\{ l_i\}_{i=1}^n$
	by the conditions: $\sum_{i=1}^{n}l_i=\alpha^+(t')$ and $l_n=l$. This gives:
\begin{align}
&n = \left\lceil \frac{\alpha^+(t') - l}{L_{\max}}\right\rceil+1, \\
&l_i = L_{\max}, ~ i= 1 \dots (n-2), \\
&l_{n-1} = \alpha^+(t') - {(n-2)L_{\max}} - l, l_n=l.
\end{align}
Then we apply the packetizer function $P^L$ and obtain $I(t)= P^L(\tilde{I}(t))$.
The packetizer function $P^L$ (Definition 1.7.3, \cite{le_boudec_network_2001}) is defined by $P^L(x) \eqdef \sup_{j \in N}\{L(j)1_{L(j)\leq x}\}$ with $L(j)\eqdef \sum_{i=1}^j {l_i}$. It transforms the bit stream $\tilde{I}(t)$ (gray line in Fig. \ref{fig:delay_curves}) into entire packets $I(t)$ (green staircase line in Fig. \ref{fig:delay_curves}).
Note that the packet of interest, i.e., packet $n$, arrives at time $t'$.

(c) For the output, we first construct the fluid output curve $F(t)$ (orange dotted-line in Fig. \ref{fig:delay_curves}) given by
\begin{align} \label{eq:output_tight}
F(t) = \inf_{0\leq s \leq t}\left\{I(s)+ \beta(t-s)\right\},
\end{align}
with $\beta(t)=\max(0, R(t-T))$, so that the service curve property in \eqref{eq:min_service_def} would be automatically satisfied if we would let $O(t)=F(t)$. However, we cannot take $O(t)=F(t)$ because $F(t)$ does not satisfy the condition that packet transmission is at rate $c$. In order to obtain the output function $O(t)$, we first observe that under the fluid output $F(t)$, packet $i$ starts transmission at $Q_i =  T+\frac{\sum_{i=1}^{i-1}{l_j}}{R}$ and finishes at $Q_i + \frac{l_i}{R}$. For the output function $O(t)$, we keep the same time $Q_i$ for the start of transmission of packet $i$, but, we let the transmission finish at time $D_i = Q_i + \frac{l_i}{c}$ ($O(t)$ is  the red line in Fig. \ref{fig:delay_curves}). Obviously, $O(t) \geq F(t) \text{ for all } t\geq 0$, therefore, the service curve property is also satisfied when the output is $O(t)$.

\textbf{Step 2.} We verify that all requirements in the theorem are satisfied. We have already shown the service curve property. Furthermore, by construction the system is FIFO, the input is packetized and packet transmission occurs at rate $c$.
	
We also need to show that the input $I(t)$ has $\alpha$ as arrival curve. Observe here that this is true by construction for the fluid input $\tilde{I}(t)$, but is less obvious for $I(t)$ since we applied the packetizer. This is where we need the assumption that the arrival curve $\alpha(t)$ is concave. We use Section 1.7.3 of \cite{le_boudec_network_2001}. Define the step function $U$ by $U(t)=\alpha^+(t')$ for all $t>0$, and $U(0)=0$. Then we have $\tilde{I}(t)=\mathcal{C}^{\alpha}\left(P^L\left(U\left(t\right)\right)\right)$, where $\mathcal{C}^{\alpha}$ is the shaper with shaping curve $\alpha$. By definition, we have $I(t)=P^L(\tilde{I}(t))=P^L\left(\mathcal{C}^{\alpha}\left(P^L\left(U\left(t\right)\right)\right)\right)$. Thus, by Theorem 1.7.2 in \cite{le_boudec_network_2001}, $I(t)$ has $\alpha$ as arrival curve.

Last, we show that packet $n$ achieves the delay bound. We have $Q_n =  T+\frac{\sum_{i=1}^{n-1}{l_j}}{R} = T + \frac{\alpha^+(t')-l}{R}$. Furthermore, $A_n = t'$ and $D_n=Q_n+\frac{l}{c}$, therefore,
\begin{align}\label{eq:queuing_tight2}
 D_n -A_n&= T+\frac{\alpha^+(t')-l}{R} + \frac{l}{c} -t'= \Delta -\frac{l}{R} + \frac{l}{c}.
\end{align}
\end{proof} 
\section{Conclusion}\label{sec:conclusion}
We considered a network element that offers a rate-latency service curve and has a known transmission rate larger than the rate guaranteed by the service curve. 
We have obtained a delay bound that uses both the information of the rate-latency service curve and the constant transmission rate. It improves on the existing network calculus bound by an amount that depends on the length of the packet being transmitted. The improvement is larger for larger packet lengths.

\bibliographystyle{IEEEtran}
\vspace{-0.05in}
\bibliography{ref,leb}

\begin{thebibliography}{10}
\providecommand{\url}[1]{#1}
\csname url@samestyle\endcsname
\providecommand{\newblock}{\relax}
\providecommand{\bibinfo}[2]{#2}
\providecommand{\BIBentrySTDinterwordspacing}{\spaceskip=0pt\relax}
\providecommand{\BIBentryALTinterwordstretchfactor}{4}
\providecommand{\BIBentryALTinterwordspacing}{\spaceskip=\fontdimen2\font plus
\BIBentryALTinterwordstretchfactor\fontdimen3\font minus
  \fontdimen4\font\relax}
\providecommand{\BIBforeignlanguage}[2]{{%
\expandafter\ifx\csname l@#1\endcsname\relax
\typeout{** WARNING: IEEEtran.bst: No hyphenation pattern has been}%
\typeout{** loaded for the language `#1'. Using the pattern for}%
\typeout{** the default language instead.}%
\else
\language=\csname l@#1\endcsname
\fi
#2}}
\providecommand{\BIBdecl}{\relax}
\BIBdecl

\bibitem{detnet}
\BIBentryALTinterwordspacing
{Deterministic Networking (detnet)}. Accessed: 2019-03-26. [Online]. Available:
  \url{https://datatracker.ietf.org/wg/detnet/about/}
\BIBentrySTDinterwordspacing

\bibitem{TSN}
\BIBentryALTinterwordspacing
{IEEE TSN}. Accessed: 2019-03-27. [Online]. Available:
  \url{https://1.ieee802.org/tsn/}
\BIBentrySTDinterwordspacing

\bibitem{le_boudec_network_2001}
J.-Y. Le~Boudec and P.~Thiran, \emph{Network Calculus: A Theory of
  Deterministic Queuing Systems for the Internet}.\hskip 1em plus 0.5em minus
  0.4em\relax Springer Science \& Business Media, 2001, vol. 2050.

\bibitem{Changbook}
C.~S. Chang, \emph{Performance Guarantees in Communication Networks}.\hskip 1em
  plus 0.5em minus 0.4em\relax New York: Springer-Verlag, 2000.

\bibitem{bouillard2018deterministic}
A.~Bouillard, M.~Boyer, and E.~Le~Corronc, \emph{Deterministic Network
  Calculus: From Theory to Practical Implementation}.\hskip 1em plus 0.5em
  minus 0.4em\relax John Wiley \& Sons, 2018.

\bibitem{GLV95}
P.~Goyal, S.~S. Lam, and H.~Vin, ``{Determining End-to-End Delay Bounds In
  Heterogeneous Networks},'' in \emph{5th Int Workshop on Network and Operating
  Systems Support for Digital Audio and Video}, Durham NH, April 1995.

\bibitem{azua_complete_2014}
J.~A. Ruiz De~Azua and M.~Boyer, ``{Complete Modelling of AVB in Network
  Calculus Framework},'' in \emph{22nd International Conference on Real-Time
  Networks and Systems}, ser. RTNS '14, Versaille, France, 2014, pp. 55--64.

\bibitem{mohammadpour_latency_2018}
E.~{Mohammadpour}, E.~{Stai}, M.~{Mohiuddin}, and J.~{Le Boudec}, ``{Latency
  and Backlog Bounds in Time-Sensitive Networking with Credit Based Shapers and
  Asynchronous Traffic Shaping},'' in \emph{30th International Teletraffic
  Congress (ITC 30)}, vol.~02, Sepember 2018, pp. 1--6.

\bibitem{daigmorte_modelling_2018}
\BIBentryALTinterwordspacing
H.~Daigmorte, M.~Boyer, and L.~Zhao, ``{Modelling in Network Calculus a {TSN}
  Architecture Mixing Time-Triggered, Credit Based Shaper and Best-Effort
  Queues},'' 2018. [Online]. Available:
  \url{https://hal.archives-ouvertes.fr/hal-01814211}
\BIBentrySTDinterwordspacing

\bibitem{zhao_timing_analysis_2018}
L.~Zhao, P.~Pop, Z.~Zheng, and Q.~Li, ``{Timing Analysis of AVB Traffic in TSN
  Networks using Network Calculus},'' in \emph{Real-Time and Embedded
  Technology and Applications Symp.}, ser. {RTAS} '18, 2018, pp. 25--36.

\bibitem{le_boudec_theory_2018}
J.-Y. Le~Boudec, ``{A Theory of Traffic Regulators for Deterministic Networks
  With Application to Interleaved Regulators},'' \emph{{IEEE}/{ACM} Trans. on
  Networking (TON)}, vol.~26, no.~6, pp. 2721--2733, 2018.

\bibitem{liebeherr_duality_2017}
J.~Liebeherr, ``{Duality of the Max-Plus and Min-Plus Network Calculus},''
  \emph{Foundations and Trends® in Netw.}, vol.~11, no.~3, pp. 139--282, 2017.

\end{thebibliography}


\end{document}